\documentclass[a4paper,11pt]{article}

\usepackage{authblk}
\usepackage[utf8]{inputenc}
\usepackage{amsthm}

\usepackage{amsmath,amssymb,enumitem}
\usepackage[linesnumbered, ruled, vlined]{algorithm2e} 

\newtheorem{theorem}{Theorem}

\newtheorem{lemma}[theorem]{Lemma}
\newtheorem{corollary}[theorem]{Corollary}
\newtheorem{definition}[theorem]{Definition}
\newtheorem{remark}{Remark}

\usepackage{tikz, pgfplots}
\usetikzlibrary{patterns}
\usetikzlibrary{shapes.geometric, calc}

\usepackage[margin=1in]{geometry}
\usepackage{hyperref}
\usepackage{color}
\definecolor{darkgreen}{RGB}{0,100,0}
\definecolor{firebrick}{RGB}{178,34,34}

\DeclareMathOperator{\Probability}{\mathbb{P}}

\DeclareMathOperator{\Expected}{\mathbb{E}}
\newcommand{\Ex}[1]{\Expected\pbrcx{#1}}
\newcommand{\Prob}[1]{\Probability\pbrcx{#1}}

\newcommand{\R}{{\mathbb{R}}}

\providecommand{\keywords}[1]{\textbf{\textit{keywords---}} #1}

\newcommand{\pth}[1]{\ensuremath{\left(#1\right)}}
\newcommand{\pbrcx}[1]{\ensuremath{\left[#1\right]}}

\newcommand{\rad}{{\rm{rad}}}

\newcommand{\e}{{\varepsilon}}

\newcommand{\Cech}{{\v{C}ech }}

\newcommand{\bary}{B}
\newcommand{\wbary}{{\hat{B}}}
\newcommand{\conv}{{\rm conv}}

\renewcommand{\wp}{\hat{P}}

\newcommand*\samethanks[1][\value{footnote}]{\footnotemark[#1]}

\title{Dimensionality Reduction for $k$-Distance Applied to Persistent Homology}

\author[1]{Shreya Arya
\thanks{
The research leading to these results has received funding from the European Research Council (ERC) under  the  European  Union's  Seventh  Framework  Programme (FP/2007-2013)  /  ERC  Grant Agreement No. 339025 GUDHI (Algorithmic Foundations of Geometry Understanding in Higher Dimensions).}
}

\author[2]{Jean-Daniel Boissonnat
\samethanks[1]
\thanks{Supported by the French government, through the 3IA C\^ote d'Azur Investments in the Future project managed by the 
National Research Agency (ANR) with the reference number ANR-19-P3IA-0002.
}
}

\author[3]{Kunal Dutta
\samethanks[1]
\thanks{Supported by the Polish NCN SONATA Grant no. 2019/35/D/ST6/04525. 
}
}

\author[4]{Martin Lotz}

\affil[1]{Duke University, Durham, USA. email: shreya.arya14@gmail.com}
\affil[2]{Universit\'e C\^ote d'Azur, INRIA, Sophia-Antipolis, France, email: jean-daniel.boissonnat@inria.fr}
\affil[3]{Faculty of Mathematics, Informatics, and Mechanics, University of Warsaw, Poland. email: K.Dutta@mimuw.edu.pl}
\affil[4]{Mathematics Institute, University of Warwick, UK. email: martin.lotz@warwick.ac.uk}

\begin{document}

\maketitle

\begin{abstract}
  Given a set $P$ of $n$ points and a constant $k$, we are interested
  in computing the persistent homology of the \Cech filtration of $P$
  for the $k$-distance, and investigate the effectiveness of
  dimensionality reduction for this problem, answering an
  open question of Sheehy [\emph{Proc. SoCG}, 2014].
  We show that \emph{any} linear transformation that preserves 
  pairwise distances up to a $(1\pm \e)$ multiplicative factor, 
  must preserve the persistent homology of the \Cech filtration 
  up to a factor of $(1-\e)^{-1}$.
  Our results also show that the 
  Vietoris-Rips and Delaunay filtrations for the $k$-distance, as well as the \Cech filtration for the 
  approximate $k$-distance of Buchet et al. [\emph{J. Comput. Geom.}, 2016] are preserved up to a $(1\pm \e)$ factor.

  We also prove extensions of our main theorem, for point sets $(i)$ lying in a region of bounded   
  Gaussian width or $(ii)$ on a low-dimensional submanifold, obtaining embeddings having the dimension bounds of Lotz 
  [\emph{Proc. Roy. Soc.}, 2019] and Clarkson [\emph{Proc. SoCG}, 2008] respectively.
  Our results also work in the \emph{terminal dimensionality reduction} setting, where the distance of any point
  in the original ambient space, to any point in $P$, needs to be approximately preserved.
\end{abstract}

 \\

	\section{Introduction}
	\label{sec:intro}
	
	Persistent homology is one of the main tools to extract information from data in topological data analysis. 
	Given a data set as a point cloud in some ambient space, the idea is to construct a filtration sequence of topological spaces from the point cloud, 
	and extract topological information from this sequence. The topological spaces are usually constructed by considering balls around the data points, 
	in some given metric of interest, as the open sets. However the usual distance function is highly sensitive to the presence of outliers and noise.  
	One approach is to use distance functions that are more robust to outliers, such as the \emph{distance-to-a-measure} and the related \emph{$k$-distance} 
	(for finite data sets), proposed recently by Chazal et al. \cite{DBLP:journals/focm/ChazalCM11}
	Although this is a promising direction, an exact implementation can have significant cost in run-time. To overcome this difficulty,  
	approximations of the $k$-distance have been proposed recently  that led to certified approximations of  persistent 
	homology \cite{DBLP:journals/dcg/GuibasMM13, DBLP:journals/comgeo/BuchetCOS16}. Other approaches involve using kernels~\cite{DBLP:conf/compgeom/PhillipsWZ15} and
	de-noising algorithms~\cite{DBLP:journals/jocg/BuchetDWW18,journals/sis/Zhang13}.
	
	In all the above settings, the sub-routines required for computing persistent homology have exponential or worse dependence on the 
	ambient dimension, and rapidly become unusable in real-time once the dimension grows beyond a few dozens - which is indeed the case in many applications, 
	for example in image processing, neuro-biological networks, and data mining (see e.g.~\cite{giraud2014introduction}). This phenomenon is often referred 
	to as the \emph{curse of dimensionality}. 
\vspace{2mm}

\noindent {\bf The Johnson-Lindenstrauss Lemma.}
	One of the simplest and most commonly used mechanisms to mitigate this curse, is that of \emph{random projections}, as applied in the
	celebrated  Johnson-Lindenstrauss lemma (JL Lemma for short)~\cite{johnson_lindenstrauss_1984}.  
	The JL Lemma states that any set of
        $n$ points in Euclidean space can be embedded into a space of
        dimension $ O(\e^{-2}\log n) $ with $(1 \pm \e) $
        distortion. 
	Since the initial non-constructive proof of this fact by Johnson and Lindenstrauss~\cite{johnson_lindenstrauss_1984}, 
	several authors have given successive improvements, e.g., Indyk, Motwani, Raghavan and Vempala~\cite{DBLP:conf/stoc/IndykMRV97}, 
	Dasgupta and Gupta~\cite{DBLP:journals/rsa/DasguptaG03}, Achlioptas~\cite{DBLP:conf/pods/Achlioptas01}, 
	Ailon and Chazelle~\cite{DBLP:journals/siamcomp/AilonC09}, Matou\v{s}ek~\cite{doi:10.1002/rsa.20218}, 
	Krahmer and Ward~\cite{DBLP:journals/siamma/KrahmerW11}, and Kane and Nelson~\cite{DBLP:journals/jacm/KaneN14}.
        These address the issues of \emph{efficient} construction
	and implementation, using random matrices that support fast multiplication. 
        Dirksen~\cite{DBLP:journals/focm/Dirksen16} gave a unified theory for 
        dimensionality reduction using subgaussian matrices.

        In a different direction, variants of the Johnson-Lindenstrauss lemma giving embeddings into spaces of lower dimension than the JL bound have been given 
        under several specific settings. For point sets lying in regions of bounded Gaussian width, a theorem of Gordon~\cite{10.1007/BFb0081737} 
        implies that the dimension of the embedding can be reduced to a function of the Gaussian width, independent of the number of 
        points. Sarlos~\cite{DBLP:conf/focs/Sarlos06} showed that points lying on a $d$-flat can be mapped to $O(d/\e^2)$ dimensions 
        independently of the number of points. Baraniuk and Wakin~\cite{DBLP:journals/focm/BaraniukW09} proved an analogous result for points on a smooth submanifold 
        of Euclidean space,
        which was subsequently sharpened by Clarkson~\cite{DBLP:conf/compgeom/Clarkson08} (see also Verma~\cite{nverma2011}), whose version directly preserves geodesic distances on 
        the submanifold. Other related results include those of Indyk and Naor~\cite{DBLP:conf/compgeom/Clarkson08} 
        for sets of bounded doubling dimension and Alon and Klartag~\cite{DBLP:conf/focs/AlonK17} for general 
        inner products, with additive error only. Recently, Nelson and Narayanan~\cite{DBLP:conf/stoc/NarayananN19}, building on earlier 
        results~\cite{DBLP:journals/tcs/ElkinFN17,DBLP:conf/stoc/MahabadiMMR18}, showed that for a given set of points or \emph{terminals}, using just 
        one extra dimension from the Johnson-Lindenstrauss bound, it is possible to achieve dimensionality reduction in a way 
        that preserves not only inter-terminal distances, but also distances between any terminal to \emph{any} point in the ambient space.
 
        \begin{remark}
        Our results are based on the notion of weighted points, and as in most applications of the JL lemma, give a reduced dimensionality typically of the 
        order of hundreds. This is very useful if the ambient dimensionality is much higher magnitude (e.g. $10^6$). 
        Moreover, some of the above-mentioned variants and generalizations such as for point sets having bounded Gaussian width  
        or lying on a lower-dimensional submanifold, the reduced dimensionality is independent of the number of input points, 
        which allows for still better reductions.  
        \end{remark}

\vspace{2mm}

\noindent {\bf Dimension Reduction and Persistent Homology.} The JL Lemma has also been used by Sheehy~\cite{DBLP:conf/compgeom/Sheehy14} 
	and Lotz~\cite{doi:10.1098/rspa.2019.0081} to reduce the complexity of computing persistent homology. Both Sheehy and Lotz show that the persistent 
	homology of a point cloud is approximately preserved under random projections~\cite{DBLP:conf/compgeom/Sheehy14,doi:10.1098/rspa.2019.0081}, up to 
	a $(1\pm \varepsilon)$ multiplicative factor, for any $\varepsilon\in [0,1]$. Sheehy proves this for an $n$-point set, whereas 
	Lotz's generalization applies to sets of bounded Gaussian width, and also implies dimensionality reductions for sets of bounded doubling dimension, 
        in terms of the \emph{spread} (ratio of the maximum to minimum interpoint distance).
	However, their techniques involve only the usual distance to a point set and therefore 
	remain sensitive to outliers and noise as mentioned earlier.  The question of adapting the method of random projections in order to reduce the complexity of 
	computing persistent homology using the $k$-distance, is therefore a natural one, and has been raised by Sheehy~\cite{DBLP:conf/compgeom/Sheehy14}, 
	who observed that \emph{``One notable distance function that is missing from this paper} [i.e.~\cite{DBLP:conf/compgeom/Sheehy14}] \emph{is the so-called 
	distance to a measure or \ldots $k$-distance \ldots it remains open whether the $k$-distance itself is $(1\pm \varepsilon)$-preserved under random projection."}

        \paragraph*{Our Contribution} 
	In this paper, we combine the method of random projections
        with the $k$-distance  
	and show its applicability in computing  persistent homology.
	It is not very hard to see that for a given
        point set $P$, the random Johnson-Lindenstrauss mapping
        preserves the pointwise $k$-distance to $P$
        (Theorem~\ref{jlkdist}). However,  this is not enough to  preserve intersections of balls at
        varying scales of the radius parameter, and thus does not suffice to
        preserve the persistent homology of \Cech filtrations, as noted by
        Sheehy~\cite{DBLP:conf/compgeom/Sheehy14} and Lotz~\cite{doi:10.1098/rspa.2019.0081}. 
        We show how the squared radius of a set of weighted points 
        can be expressed as a convex combination of pairwise squared distances.
        From this, it follows that the \Cech filtration under the $k$-distance, will be preserved by 
        \emph{any} linear mapping that preserves pairwise distances.

        \paragraph*{Extensions} 
        Further, as our main result applies to any linear mapping that approximately preserves pairwise distances,
        the analogous versions for bounded Gaussian width, points on submanifolds of $\R^D$, terminal dimensionality reduction and others apply immediately. 
        Thus, we give several extensions of our results. The extensions provide bounds which do not depend on the number of points in the sample. 
        The first one, analogous to~\cite{doi:10.1098/rspa.2019.0081}, shows that the persistent homology with 
        respect to the $k$-distance, of point sets contained in regions having bounded Gaussian width, can be 
        preserved via dimensionality reduction, using an embedding with dimension bounded by a function of the Gaussian width.
        Another result is that for points lying in a low-dimensional submanifold of a high-dimensional Euclidean space, 
        the dimension of the embedding preserving the persistent homology with $k$-distance depends linearly on the dimension of the submanifold.
        Both these settings are commonly encountered in high-dimensional data analysis and machine learning (see, e.g., the 
        \emph{manifold hypothesis}~\cite{fefferman2016testing}). We mention that analogous to~\cite{DBLP:conf/stoc/NarayananN19}, 
        it is possible to preserve the $k$-distance based persistent homology while also preserving the distance  
        from any point in the ambient space to every point (i.e., terminal) in $P$ (and therefore the $k$-distance to $P$), using just one extra dimension.
        
        \paragraph{Run-time and Efficiency} 
        In many other applications of the Johnson-Lindenstrauss dimensionality reduction, multiplying by a dense gaussian matrix 
        is a significant overhead, and can seriously affect any gains resulting from working in a lower dimensional space.
        However, as is pointed out in~\cite{doi:10.1098/rspa.2019.0081}, in the computation of persistent homology the dimensionality 
        reduction step is carried out only once for the $n$ data points at the beginning of the construction. 
        Having said that, it should still be observed that most of the recent results on dimensionality reduction using sparse subgaussian matrices 
        ~\cite{DBLP:journals/siamcomp/AilonC09,DBLP:journals/jacm/KaneN14,DBLP:journals/siamma/KrahmerW11}
        can also be used to compute the $k$-distance persistent homology, with little to no extra cost.

	\begin{remark} It should be noted that the approach of using dimensionality reduction for the $k$-distance, is complementary 
	to denoising techniques such as~\cite{DBLP:journals/jocg/BuchetDWW18} as we do not try to remove noise, only 
	to be more robust to noise. Therefore, it can be used in conjunction with 
	denoising techniques, as a pre-processing tool when the dimensionality is high. 
        \end{remark}

        \paragraph*{Outline} 
	 The rest of this paper is organized as follows. In
        Section~\ref{sec:background}, we briefly summarize some basic
	definitions and background. Our theorems are stated and proven in
        Section~\ref{sec: results}. 
        Some applications of our results are derived in Section~\ref{sec:appln}.
	We end with some final remarks and open questions in
        Section~\ref{sec:future-work}.

\section{Preliminaries}
\label{sec:background}
      We need a well-known identity for the variance of bounded random variables, which will be crucial in the proof of our main theorem. A 
      short probabilistic proof of~\eqref{eqn:basic-var-id} is given in the Appendix.
          Let $A$ be a set of points $p_1,\ldots,p_l\in \R^m$. A point $b\in \R^m$ is a \emph{convex combination} of the points in $A$ if
          there exist non-negative reals $\lambda_1,\ldots,\lambda_l\geq 0$ such that $b=\sum_{i=1}^l \lambda_i p_i$ and $\sum_{i=1}^l \lambda_i=1$.

          Let $b = \sum_{i=1}^k \lambda_i p_i$ be a convex combination of points $p_1,\ldots,p_k \in \R^D$. Then for any point $x\in \R^D$,
          \begin{eqnarray}
             \sum_{i=1}^k \lambda_i \|x-p_i\|^2  &=& \|x-b\|^2 + \sum_{i=1}^k \lambda_i \|b-p_i\|^2. \label{eqn:basic-var-id}
          \end{eqnarray}
          In particular, if $\lambda_i = 1/k$ for all $i$, we have 
          \begin{eqnarray}
             \frac{1}{k}\sum_{i=1}^k \|x-p_i\|^2  &=& \|x-b\|^2 + \sum_{i=1}^k \frac{1}{k} \|b-p_i\|^2. \label{eqn:basic-var-id2}
          \end{eqnarray}

	\subsection{The Johnson-Lindenstrauss Lemma}
	\label{subsec:rand-proj}

	The Johnson-Lindenstrauss Lemma
        \cite{johnson_lindenstrauss_1984} states that any subset of
        $n$ points of Euclidean space can be embedded in a space of
        dimension $ O(\e^{-2}\log n) $ with $ (1 \pm \e) $
        distortion. 
	We use the notion of an \emph{$\e$-distortion map with respect to $P$} (also commonly called a Johnson-Lindenstrauss map).
        \begin{definition}
            Given a point set $P\subset \R^D$, and $\e\in (0,1)$, a mapping $f:\R^D\to \R^d$ for some 
        $d\leq D$ is an \emph{$\e$-distortion map with respect to $P$}, if for all $x,y\in P$, 
        \[ (1-\e)\|x-y\| \leq \|f(x)-f(y)\| \leq (1+\e)\|x-y\|.\]
        \end{definition}

        A random variable $X$ with mean zero is said to be \emph{subgaussian} with \emph{subgaussian norm} $K$ if 
        $\Ex{\exp\pth{X^2/K^2}} \leq 2$. In this case, the tails of the random variable satisfy
        \[ \Prob{|X|\geq t} \leq 2\exp\pth{-t^2/2K^2}.\]
        We focus on the case where the Johnson-Lindenstrauss embedding is carried out via random subgaussian matrices, 
        i.e., matrices where for some given $K >0$, each entry is an independent subgaussian random variable with subgaussian norm $K$.
	This case is general enough to include the mappings of Achlioptas~\cite{DBLP:conf/pods/Achlioptas01}, 
	Ailon and Chazelle~\cite{DBLP:journals/siamcomp/AilonC09}, Dasgupta and Gupta~\cite{DBLP:journals/rsa/DasguptaG03}, 
        Indyk, Motwani, Raghavan, and Vempala~\cite{DBLP:conf/stoc/IndykMRV97}, and Matou\v{s}ek~\cite{doi:10.1002/rsa.20218}
	(see Dirksen for a unified treatment~\cite{DBLP:journals/focm/Dirksen16}). 

	\begin{lemma}[JL Lemma]
		\label{JLlemma}
		Given $ 0 < \e,\delta < 1 $, and a finite point set $ P
                \subset \mathbb{R}^D$ of size $|P|=n$.
		Then a random linear mapping $ f \colon \mathbb{R}^D \to \mathbb{R}^d $  
		where $ d=O(\e^{-2}\log n)  $ given by $f(v) =
                \sqrt{\frac{D}{d}}Gv$ 
                where $G$ is a $d\times D$ subgaussian random matrix,
                is an $\e$-distortion map with respect to $P$, with probability at least $1-\delta$.
	\end{lemma}
    
        \begin{definition}
        \label{def:eps-distort-map}
        For ease of recall, we shall refer to a random linear 
        mapping $ f \colon \mathbb{R}^D \to \mathbb{R}^d $  
                given by $f(v) =
                \sqrt{\frac{D}{d}}Gv$ 
                where $G$ is a $d\times D$ subgaussian random matrix, as a 
        \emph{subgaussian $\e$-distortion map}.
        \end{definition}

While in the version given here the dimension of the embedding depends on the number of points in $P$ and subgaussian projections, 
the JL lemma has been generalized and extended in several different directions, some of which are briefly outlined below. The 
generalization of the results of this paper to these more general settings is straightforward.

        \paragraph{Sets of Bounded Gaussian Width}

        \begin{definition}
           Given a set $S \subset \R^D$, the \emph{Gaussian width} of $S$ is
           \[ w(S) := \Ex{\sup_{x\in S} \langle x,g\rangle},\]
           where $g \in \R^D$ is a random standard $D$-dimensional Gaussian vector. 
        \end{definition}

        In several areas like geometric functional analysis, compressed sensing, machine learning, etc. the Gaussian width 
        is a very useful measure of the width of a set in Euclidean space (see e.g.~\cite{DBLP:books/daglib/0036092} 
        and the references therein). It is also closely related to the \emph{statistical dimension} of a set (see e.g. 
        ~\cite[Chapter 7]{vershynin_2018}).
        The following analogue of the Johnson Lindenstrauss lemma for sets of bounded Gaussian width was
        given in~\cite{doi:10.1098/rspa.2019.0081}. It essentially follows from a result of Gordon~\cite{10.1007/BFb0081737}.  

        \begin{theorem}[~\cite{doi:10.1098/rspa.2019.0081}, Theorem 3.1]
        \label{thm:lotz-gordon-gw}
        Given $\e,\; \delta \in (0,1)$, $P\subset \R^D$, let $S := \{(x-y)/\|x-y\| \;:\; x,y \in P\}$. Then 
        for any $d\geq \frac{\pth{w(S)+\sqrt{2\log (2/\delta)}}^2}{\e^2}+1$, 
        the function $f:\R^D\to \R^d$ given by $f(x) = \pth{\sqrt{D/d}}Gx$, where $G$ is a random $d\times D$ Gaussian matrix $G$, 
        is 
        a subgaussian $\e$-distortion map with respect to $P$, with probability at least $1-\delta$.
        \end{theorem}

The result extends to subgaussian matrices with slightly worse constants. One of the benefits of this version is that the set $P$ does not need to be finite. We refer 
to~\cite{doi:10.1098/rspa.2019.0081} for more on the Gaussian width in our context. 
        \paragraph{Submanifolds of Euclidean Space}
        \label{sec:ld-manif}
        For point sets lying on a low-dimensional submanifold of a high-dimensional Euclidean space, one can obtain an embedding with a smaller dimension using the 
        bounds of Baraniuk and Wakin~\cite{DBLP:journals/focm/BaraniukW09} or Clarkson~\cite{DBLP:conf/compgeom/Clarkson08}, which will depend only on the parameters of the submanifold.   
       %
        Clarkson's theorem is summarised below.
        \begin{theorem}[Clarkson~\cite{DBLP:conf/compgeom/Clarkson08}]
        \label{thm:clark-manif-bdd}
        There exists an absolute constant $c>0$ such that, given a connected, compact, orientable, differentiable $\mu$-dimensional 
        submanifold $M \subset \R^D$, and $\e,\delta \in (0,1)$,  
        a random projection map $f:\R^D\to \R^d$, given by $v\mapsto \sqrt{\frac{D}{d}}Gv$, where $G$ is a $d\times D$ subgaussian random matrix, 
        is an $\e$-distortion map with respect to $P$, with probability at least $1-\delta$, for 
        \[ d \geq c\pth{\frac{\mu\log(1/\e)+\log(1/\delta)}{\e^2} + \frac{C(M)}{\e^2}},\]
        where $C(M)$ depends only on $M$. 
        \end{theorem}


        \paragraph{Terminal Dimensionality Reduction}
        In a recent breakthrough result, Narayanan and Nelson~\cite{DBLP:conf/stoc/NarayananN19} showed that it is possible to 
        $(1\pm O(\e))$-preserve distances from a set of $n$ \emph{terminals} in a high-dimensional space to \emph{every point} in the 
        space, using only one dimension more than the Johnson-Lindenstrauss bound. 
        A summarized version of their theorem is as follows. 
        The derivation of the second statement is given in the Appendix.
        \begin{theorem}[\cite{DBLP:conf/stoc/NarayananN19}, Theorem 3.2, Lemma 3.2]
        \label{thm:nel-nar-term-dim-red}
        Given terminals $x_1,\ldots,x_n\in \R^D$ and $\e\in (0,1)$, there exists a non-linear map
        $f:\R^D\to \R^{d'}$ with $d'=d+1$, where $d= O\pth{\frac{\log n}{\e^2}}$ is the bound given in Lemma~\ref{JLlemma}, 
        such that $f$ is an $\e$-distortion map for any pairwise distance between $x_i,x_j\in P$,
        and an $O(\e)$-distortion map for the distances between any pairs of points $(x,u)$, where $x\in P$ and $u\in \R^D$. 
        Further, the projection of $f$ to its first $d-1$ coordinates is a subgaussian $\e$-distortion map.
        \end{theorem}
        As noted in~\cite{DBLP:conf/stoc/NarayananN19}, any such map must necessarily be non-linear. Suppose not, 
        then on translating the origin to be a terminal, it follows that the Euclidean norm of each point on the unit sphere around the origin 
        must be $O(\e)$-preserved, which means that the dimension of any embedding given by a linear map would not be any less than the original dimension.

\vspace{2mm}

\subsection{$k$-Distance }
\label{subsec:dtmbackground}
The distance to a finite point set $ P $ is usually taken to be the minimum distance to a point in the set. 
For the computations involved in geometric and topological inference, however, this distance is highly sensitive to outliers and noise.
To handle this problem of sensitivity, Chazal et al. in
\cite{DBLP:journals/focm/ChazalCM11} introduced the \emph{distance to a
probability measure} which, in the case of a uniform probability on
$P$, is called the \emph{$k$-distance}.

	\begin{definition}[$k$-distance]
	For $ k \in\{1,...,n\} $ and $ x \in \mathbb{R}^D $, the $ k
        $-distance of $x$ to $P$ is
		\begin{equation}
		d_{P,k}(x)= \min_{S_k\in
                  {P\choose k }} \sqrt{\dfrac{1}{k}\sum_{p \in
                    S_k}\|x-p\|^2}=\sqrt{\dfrac{1}{k}\sum_{p \in
                    \text{NN}_{P}^k(x)}\|x-p\|^2} 
		\end{equation} 
		where $ \text{NN}^{k}_{P}(x) \subset P $ denotes the $
                k $ nearest neighbours in $ P $ to the point $ x \in
                \mathbb{R}^{D} $.
\label{def:k-distance}
	\end{definition}

It was shown in~\cite{DBLP:journals/dcg/Aurenhammer90}, that the $ k $-distance can be
expressed in terms of weighted points and \emph{power distance}. 
A weighted point $\hat{p}$ is a point $p$ of $\R^D$ together with a (not necessarily positive) 
real number called its weight and denoted by $w(p)$.
The \emph{power distance} between a point $x\in \R^D$ and a weighted point $\hat{p}=(p,w(p))$, denoted by 
$D(x,\hat{p})$ is $\|x-p\|^2-w(p)$, i.e. the power of $x$ with respect to a ball of radius $\sqrt{w(p)}$ centered at $p$.
The distance between two weighted points $\hat{p}_i=(p_i,w(i))$ and
  $\hat{p}_j=(p_j,w(j))$ is defined as
  $D(\hat{p}_i,\hat{p}_j)=\|p_i-p_j\|^2 - w(i)- w(j)$.  This definition
  encompasses the case where the two weights
  are 0, in which case we have the squared Euclidean distance, and the case where one of the points has weight 0, in which
  case, we have the power distance of a point to a ball.  We say
  that two weighted points are {\em orthogonal} when
  their weighted distance is zero.

\noindent Let $ \bary_{P,k} $ be the set of iso-barycentres of all subsets
of $ k $ points in $ P $. To each barycenter $b\in \bary_{P,k}$, $b=
(1/k) \sum_{i}p_{i} $, we associate the weight $ w(b)=-
\frac{1}{k} \sum_{i}\|b-p_i\|^2 $. Note that, despite the notation, this weight does not only depend on $b$, but also on the set of points in $P$ for which 
$b$ is the barycenter. 
Writing $\wbary_{P,k}= \{ \hat{b}=(b, w(b)), b\in \bary_{P,k}\}$, we see from~\eqref{eqn:basic-var-id2} that the $k$-distance is 
the square root of a
power distance~\cite{DBLP:journals/dcg/Aurenhammer90}
	\begin{equation} \label{eq:dppx-power-dist-defn}
	d_{P,k}(x) = \min_{\hat{b}\in \wbary_{P,k}} \sqrt{D(x,\hat{b})}.
	\end{equation} 
	Observe that in general the squared distance between a pair of weighted points can be negative, but 
	the above assignment of weights ensures that the $k$-distance $d_{P,k}$ is a real function. 
	Since $ d_{P,k} $ is the square root of a
        non-negative power distance, the $\alpha$-sublevel set of $ d_{P,k} $,
        ${d}_{P,k}([-\infty, \alpha])$, $\alpha\in \R$,  is the
        union of $ n\choose k $ balls 
        $B(b, \sqrt{\alpha^2 + w(b)})$,
        $b\in \bary_{P,k}$.
	However, some of the balls may  be included  in the union of
        others and be redundant. In fact, the number of  barycenters (or equivalently of balls) required
        to define a level set of $ d_{P,k} $ is equal to the number of
        the non-empty cells in the $ k $th-order 
	Voronoi diagram of $P$. Hence the number of non-empty cells is 
	$ \Omega \left( n^{\left \lfloor (D+1)/2 \right  \rfloor}
\right) $
        \cite{DBLP:journals/dcg/ClarksonS89} and computing them 
in high dimensions is intractable. It is then natural to look for
approximations of the $k$-distance, as 
	proposed  in \cite{DBLP:journals/comgeo/BuchetCOS16}.
	
	\begin{definition}[Approximation] Let $ P \subset
    \mathbb{R}^{D} $ and $ x\in \mathbb{R}^{D} $. The  approximate
    $k$-distance $ \tilde{d}_{P,k}(x) $ is defined as
		\begin{eqnarray}
			\tilde{d}_{P,k}(x) &:=& \min _{{p}\in {P}}\sqrt{ D(x,\hat{p})    }
\label{eq:dppx-approx}
		\end{eqnarray}
where $\hat{p}=(p,w(p))$ with $w(p)= -d^2_{P,k}(p) $, the negative of
the squared $ k $-distance of $ p $.
\label{def:approximation}
	\end{definition}

In other words, we replace the set of barycenters with $P$. 
As in the exact case, $\tilde{d}_{P,k}$ is  the square root of a power
distance and its $\alpha$-sublevel set, $\alpha\in \R$, is a union of balls,
specifically  the balls  $B(p,
\sqrt{\alpha^2-d_{P,k}^2(p)})$, $p\in P$. The major difference with the
exact case is that, since we consider only balls around the points
of $P$, their number is $n$ instead of $ n\choose k $ in the
exact case (compare Eq.~(\ref{eq:dppx-approx})  and Eq.~\eqref{eq:dppx-power-dist-defn}).
Still, $\tilde{d}_{P,k}(x)$
        approximates the $k$-distance  
        \cite{DBLP:journals/comgeo/BuchetCOS16}:
	\begin{eqnarray}
	    \dfrac{1}{\sqrt{2}} \ d_{P,k} \le \tilde{d}_{P,k} \le \sqrt{3} \ d_{P,k}. \label{eqn:approx-k-dist-approx} 
	\end{eqnarray}
	
        We now make an observation for the case when the weighted points are barycenters, which will be useful in proving our main theorem.
	\begin{lemma}
        \label{l:pow-dist-pairwise-dist}
        If $b_1,b_2 \in \bary_{P,k}$, and $p_{i,1},\ldots,p_{i,k} \in P$ for $i=1,2$, such that $b_i = \frac{1}{k}\sum_{l=1}^k p_{i,l}$, 
        and $w(b_i) = \frac{1}{k}\sum_{l=1}^k\|b_i-p_{i,l}\|^2$ for $i=1,2$, then
            \[D(\hat{b}_1,\hat{b}_2) = \dfrac{1}{k^2}\sum_{l,s=1}^k \|p_{1,l}-p_{2,s}\|^2.\]
	\end{lemma}

        \begin{proof}
        We have 
        \[D(\hat{b}_1,\hat{b}_2) \;\;=\;\; \|b_1-b_2\|^2-w(b_1)-w(b_2) 
                                 \;\;=\;\; \|b_1-b_2\|^2+\dfrac{1}{k}\sum_{l=1}^k \|b_1-p_{1,l}\|^2+\dfrac{1}{k}\sum_{l=1}^k \|b_2-p_{2,l}\|^2.\]
        Applying the identity~\eqref{eqn:basic-var-id2}, 
        we get $\|b_1-b_2\|^2  +\dfrac{1}{k}\sum_{l=1}^k \|b_2-p_{2,l}\|^2 = \dfrac{1}{k}\sum_{l=1}^k\|b_1-p_{2,l}\|^2$, so 
        that 
        \begin{eqnarray}
           D(\hat{b}_1,\hat{b}_2) &=& \dfrac{1}{k}\sum_{l=1}^k \|b_1-p_{2,l}\|^2 + \dfrac{1}{k}\sum_{l=1}^k \|b_1-p_{1,l}\|^2 \notag \\
           &=& \dfrac{1}{k}\sum_{l=1}^k \|b_1-p_{2,l}\|^2 + \dfrac{1}{k^2}\sum_{s=1}^k\sum_{l=1}^k \|b_1-p_{1,l}\|^2 \notag \\
           &=& \dfrac{1}{k}\sum_{l=1}^k \pth{\|b_1-p_{2,l}\|^2 + \dfrac{1}{k}\sum_{s=1}^k \|b_1-p_{1,s}\|^2} \notag \\
           &=& \dfrac{1}{k}\sum_{l=1}^k \pth{\dfrac{1}{k}\sum_{s=1}^k\|p_{1,s}-p_{2,l}\|^2} \;\;=\;\;   
               \dfrac{1}{k^2}\sum_{l,s=1}^k \|p_{1,s}-p_{2,l}\|^2,  \label{eqn:var-ineq-bi}
        \end{eqnarray}
        where in~\eqref{eqn:var-ineq-bi}, we again applied ~\eqref{eqn:basic-var-id2} to each of the points $p_{2,s}$, with respect to the barycenter $b_1$.
        \end{proof}

	\subsection{Persistent Homology}
	\label{subsec:ph}

\noindent {\bf Simplicial Complexes and Filtrations} 
Let $V$ be a finite set. 
      An (abstract) simplicial complex with vertex set $V$ is a set $ K $ of finite
      subsets of $V$  such that if $ A \in K $ and $ B \subseteq A$,
then $ B \in K $. The sets in $ K $ are called
      the simplices of $ K $. A simplex $F \in K$ that is strictly contained in a simplex 
      $A\in K$, is said to be a \emph{face} of $A$.

	A simplicial complex $ K $ with a function $ f: K \to
        \mathbb{R} $ such that $ f(\sigma) \le f(\tau) $ whenever
        $\sigma$ is a face of $\tau$ is a filtered simplicial
        complex. The sublevel set of $f$ at $ r \in \mathbb{R}$, $
        \mathnormal{f}^{-1}\left(-\infty,r \right] $, is a subcomplex of $ K $. 
	By considering different values of $ r $, we get a nested
        sequence of subcomplexes (called a filtration) of $ K $, $ \emptyset= K^0\subseteq K^1 \subseteq ... \subseteq K^m=K $, where $ K^{i} $ is the sublevel set at value $ r_i $. 

The        \Cech filtration associated to  a finite set $P$ of points
in $\mathbb{R}^D $ plays an important role in Topological Data Analysis.

	\begin{definition}[\Cech Complex]
 The \Cech complex $\check{C}_\alpha(P)$ is the set of simplices $\sigma\subset P$ such that $\rad(\sigma$) $\le$ $\alpha$, where $\rad(\sigma)$ is the radius of the smallest 
	enclosing ball of $ \sigma $, i.e.
\[\rad(\sigma) = \min_{x\in\R^D} \max_{p_i \in \sigma} \|x-p_i\| .\]
\end{definition}
When the threshold $\alpha$ goes from $0$ to $+\infty$, we obtain the \Cech
filtration of $P$.
$\check{C}_\alpha(P)$ can be equivalently defined as the nerve  of the
closed balls
          $\overline{B}(p,\alpha)$, centered at the points in $P$ and
          of radius $\alpha$:  \[ \check{C}_\alpha(P)
          = \{ \sigma \subset P | \cap_{p \in
            \sigma}\overline{B}(p,\alpha) \neq \emptyset \}. \] 
	By the nerve lemma, we know that the union of balls
        $U_\alpha =\cup_{p\in P} \overline{B}(p,\alpha) $, 
	and $ \check{C}_\alpha(P) $ have the same homotopy type.

\vspace{2mm}

\noindent {\bf Persistence Diagrams.} 
Persistent homology is a means to compute and record the changes in the topology of the filtered complexes as 
the parameter $\alpha$ increases from zero to infinity. Edelsbrunner, Letscher and Zomorodian~\cite{DBLP:journals/dcg/EdelsbrunnerLZ02} 
gave an algorithm to compute the persistent homology, which takes a filtered simplicial complex as input, and 
outputs a sequence $(\alpha_{birth},\alpha_{death})$ of pairs of real numbers. Each such pair corresponds to a topological feature, and records the values 
of $\alpha$ at which the feature appears and disappears, respectively, in the filtration. Thus the topological features 
of the filtration can be represented using this sequence of pairs, which can be represented either as points in the 
extended plane $\bar{\R}^2 = \pth{\R\cup \{-\infty,\infty\}}^2$, 
called the \emph{persistence diagram}, or as a sequence of barcodes (the \emph{persistence barcode}) (see, e.g.,~\cite{DBLP:books/daglib/0025666}).
A pair of persistence diagrams $\mathbb{G}$ and $\mathbb{H}$ corresponding to the filtrations $(G_\alpha)$ and $(H_\alpha)$ respectively,
are \emph{multiplicatively $\beta$-interleaved}, $(\beta \geq 1)$, if for all $\alpha$, we have that
$G_{\alpha/\beta}  \subseteq H_{\alpha} \subseteq G_{\alpha\beta}$. We shall crucially rely on the fact that a given 
persistence diagram is closely approximated by another one if they are multiplicatively $c$-interleaved, with $c$ close to $1$ 
(see e.g.~\cite{DBLP:books/daglib/0039900}).

	The Persistent Nerve Lemma \cite{DBLP:conf/compgeom/ChazalO08} shows that the persistent homology of the \Cech complex is the 
	same as the homology of the $ \alpha $-sublevel filtration of the distance function.

\vspace{2mm}

\noindent {\bf The Weighted Case.}
        Our goal is to extend the above definitions and results to the case of the
        $k$-distance. As we
        observed earlier, the $k$-distance is a 
	power distance in
        disguise. Accordingly, we need to extend the definition of the
        \Cech complex to sets of weighted points.

	\begin{definition}[Weighted \Cech Complex]
          Let $\wp = \{ \hat{p}_1,...,\hat{p}_n\}$  be a set of weighted points, where
          $\hat{p}_i=(p_i,w(i))$. The
          $\alpha$-\Cech complex of $\wp$, $ \check{C}_\alpha(\wp)$,
          is  the  set of all simplices $\sigma$ satisfying
\[
         \exists x, \; \forall
           p_i\in \sigma, \; \|x-p_i\|^2 \leq
           w(i)+\alpha^2 \;\;\; \Leftrightarrow \;\;\; \exists x, \; \forall
           p_i\in \sigma, \; D(x,\hat{p}_i)
           \leq \alpha^2.\]
In other words,  
 the       $\alpha$-\Cech complex of $\wp$ is  the nerve of the closed
           balls
           $\overline{B}(p_i, r_i^2=w(i)+\alpha ^2)$, centered at the
           $p_i$ and of squared radius $w(i)+\alpha ^2$ (if negative,
           $\overline{B}(p_i, r_i^2)$ is
           imaginary). 
         \end{definition}

        The notions of weighted \Cech filtrations and their persistent homology now follow naturally.
Moreover, it follows from ~\eqref{eq:dppx-power-dist-defn} that the \Cech complex $
                \check{C}_{\alpha}(P)$ for the $k$-distance  is
                identical to  the weighted \Cech complex $
                \check{C}_{\alpha}(\wbary_{P,k}) $, where
                $\wbary_{P,k}$ is, as above, the set of
                iso-barycenters of all subsets of $k$ points in $P$.

	In the Euclidean case, we equivalently defined the $\alpha$-\Cech complex as the collection of simplices
        whose smallest enclosing balls have radius at most
        $\alpha$. We can proceed similarly in the weighted case.
Let $\hat{X}\subseteq \hat{P}$. We define the  {\em squared radius of   $\hat{X}$} as
\[\rad ^2 (\hat{X}) = \min_{x\in \R^{D}} \max_{\hat{p}_i\in \hat{X}} D({x},\hat{p}_i),\]
and the weighted center or simply the \emph{center} of $\hat{X}$  as the point, noted $c (\hat{X})$, where the
minimum is reached. 

Our goal is to show that preserving smallest enclosing balls in the weighted scenario under a given mapping, also preserves the 
persistent homology. Sheehy~\cite{DBLP:conf/compgeom/Sheehy14} and Lotz~\cite{doi:10.1098/rspa.2019.0081}, 
proved this for the unweighted case. Their proofs also work for the
weighted case but only under the assumption that 
the weights stay unchanged under the mapping. In our case however, the weights need to be recomputed in $f(\hat P)$. 
We therefore need a version of~\cite[Lemma 2.2]{doi:10.1098/rspa.2019.0081} for the weighted case which 
does not assume that the weights stay the same under $f$. This is Lemma~\ref{lotzlemmawtd}, which follows at the end of this 
section. 
The following lemmas will be instrumental in proving Lemma~\ref{lotzlemmawtd} and in proving our main result. 
Let $\hat{X}\subseteq
\hat{P}$ and assume without loss of generality that $\hat{X}= \{ \hat{p}_1,...,\hat{p}_m\}$,
where $\hat{p}_i=(p_i,w(i))$.

	\begin{lemma} \label{l:meb-uniq}
               $c(\hat{X})$ and $\rad(\hat{X})$ are uniquely defined.
        \end{lemma}

\begin{proof}[Proof of Lemma~\ref{l:meb-uniq}] 
The proof follows from the convexity of $D$ (see Lemma~\ref{l:pow-dist-pairwise-dist}).
Assume, for a contradiction, that there exists two centers $c_0$ and
$c_1\neq c_0$ for $\hat{X}$. For convenience, write $r=
\rad(\hat{X})$. By the definition of the center of $\hat{X}$, we have
\begin{eqnarray*}
 \exists \hat{p}_0 , \forall \hat{p}_i : D(c_0, \hat{p}_i) & \leq &D(c_0, \hat{p}_0) = \| c_0-p_0\|^2 - w(0) =
r^2 \\
 \exists \hat{p}_1, \forall \hat{p}_i : D(c_1, \hat{p}_i) & \leq & D(c_1, \hat{p}_1) = \| c_1-p_1\|^2 - w(1) =
r^2. 
\end{eqnarray*}
Consider $D_{\lambda} (\hat{p}_i)=(1-\lambda)D(c_0,\hat{p}_i) + \lambda
D(c_1,\hat{p}_i)$ and  write $c_{\lambda}
=(1-\lambda ) c_0 + \lambda c_1 $. For any $\lambda \in (0,1)$, we have
\begin{eqnarray*}
D_{\lambda} (\hat{p}_i) & = & (1-\lambda)D(c_0,\hat{p}_i) + \lambda
D(c_1,\hat{p}_i)\\
& = & (1-\lambda)(c_0-p_i)^2 + \lambda (c_1-p_i)^2 -w(i) \\
& = & D(c_{\lambda},\hat{p}_i) - c_{\lambda}^2 +(1-\lambda)c_0^2
+\lambda c_1^2 \\
& = & D(c_{\lambda},\hat{p}_i)  + \lambda (1-\lambda) (c_0-c_1)^2\\
& > & D(c_{\lambda},\hat{p}_i).
\end{eqnarray*}

Moreover, for any $i$,
\[
D_{\lambda} (\hat{p}_i)=(1-\lambda)D(c_0,\hat{p}_i) + \lambda
D(c_1,\hat{p}_i) \leq r^2.\]

Thus, for any $i$  and any $\lambda\in (0,1)$,  $D
(c_{\lambda},\hat{p}_i) < r^2$. Hence
$c_{\lambda}$ is a better center than $c_0$ and $c_1$, and $r$ is not the minimal possible value for $\rad (\hat{X})$. We have obtained a contradiction.
\end{proof}

	\begin{lemma} \label{l:meb-conv-comb}
                  Let $I$ be the set of indices for which
                  $D(c,\hat{p}_i) = \rad^2(\hat{X})$ and let
                  $\hat{X}_I=\{\hat{p}_i, i\in I\}$.
Then there exist $(\lambda_i > 0)_{i\in I}$ such that $c({\hat{X}})= \sum_{i\in I}\lambda_i{p}_i $ with $ \sum_{i\in I}\lambda_i=1 $. 
		
	\end{lemma}

\begin{proof}[Proof of Lemma~\ref{l:meb-conv-comb}]
 We write for convenience $c=c(\hat{X})$ 
  and $r=\rad(\hat{X})$ and prove that $c\in \conv (X_I)$ by contradiction.  Let  $c'\neq
c$ be the point of $\conv
(X_I)$ closest to $c$, and $\tilde{c}\neq c$ be a point on $[cc']$.  Since $\| \tilde{c}-p_i\| < \| c-p_i\|$ for all $i\in I$, 
 $D(\tilde{c},\hat{p}_i) < D(c, \hat{p}_i)$ for all $i\in I$. 
For
 $\tilde{c}$ sufficiently close to $c$, $\tilde{c}$ remains closer to
 the weighted points $\hat{p}_j$, $j\not\in I$, than to the
 $\hat{p}_i$, $i\in I$. We thus  have
 \[ D(\tilde{c},\hat{p}_j) < D(\tilde{c},\hat{p}_i) < D(c,
 \hat{p}_i)=r^2.\]
It  follows that $c$ is not the center of $\hat{X}$, a contradiction.

\end{proof}

	Combining the above results with~\cite[Lemma 4.2]{doi:10.1098/rspa.2019.0081} gives the following lemma.
	\begin{lemma} \label{l:meb-rad-conv-comb}
         Let $I$, $(\lambda_i)_{i\in I}$ be as in Lemma~\ref{l:meb-conv-comb}. Then the following holds.
		\[\rad^2(\hat{X}) =
                \dfrac{1}{2}\sum_{i\in I}\sum_{j\in I} \lambda_i\lambda_j D(\hat{p}_i,\hat{p}_j).\]
	\end{lemma}

\begin{proof}[Proof of Lemma~\ref{l:meb-rad-conv-comb}]
		From Lemma~\ref{l:meb-conv-comb}, and writing
                $c=c(\hat{X})$ for convenience, we have
		\[ \rad^2(\hat{X})= \sum_{i\in I}\lambda_i\big(\|c-p_i\|^2 - w(i)\big).\]
		We use the following simple fact from~\cite[Lemma 4.5]{doi:10.1098/rspa.2019.0081} (a probabilistic 
                proof is included in the Appendix, Lemma~\ref{l:simp-fact}). 
		\[ \sum_{i\in I}\lambda_i\|c-p_i\|^2 = \dfrac{1}{2}\sum_{i\in I}\sum_{j\in I}\lambda_i\lambda_j\|p_i-p_j\|^2. \]
		Substituting in the expression for $\rad^2(\hat{X})$, 
		\begin{eqnarray}
			\rad^2(\hat{X})   
				     & = &\dfrac{1}{2}\sum_{j\in I}\sum_{i\in I}\lambda_j \lambda_i\|p_i-p_j\|^2 - \dfrac{1}{2}\sum_{i\in I}2\lambda_iw(i) 
				                              \notag \\ 
				     & = &\dfrac{1}{2}\sum_{i,j\in I}\lambda_j \lambda_i\|p_i-p_j\|^2 - \dfrac{1}{2}\sum_{i,j\in I}2\lambda_i\lambda_j w(i) 
				          \;\; \mbox{(since }\; \sum_{j\in I} \lambda_j=1)                        \notag \\ 
				     & = &\dfrac{1}{2}\sum_{i,j\in I}\lambda_j \lambda_i\|p_i-p_j\|^2 - \dfrac{1}{2}\sum_{i,j\in I}\lambda_i\lambda_j (w(i) + w(j))
				                                   \notag \\ 
				     & = &\dfrac{1}{2}\sum_{i,j\in I}\lambda_i\lambda_j\pth{\|p_i-p_j\|^2 - w(i)-w(j) } 
                                     \;\; =\;\; \dfrac{1}{2}\sum_{i,j\in I}\lambda_i\lambda_j D(\hat{p}_i,\hat{p}_j) \notag.
		\end{eqnarray}
\end{proof}

Let $X\in \R^D$ be a finite set of points and $\hat{X}$ be the
associated weighted points where   
the weights are computed according to 
a weighting function $w: X \to \R^-$. Given a
mapping $f:\R^D\to \R^d$, we define $\widehat{f(X)}$ as the set of
weighted points $\{ (f(x), w(f(x))), x\in X\}$. Note that the weights
are recomputed in the image space $\R^d$. 
	\begin{lemma}
		\label{lotzlemmawtd}
		In the above setting, if $ f$ is such that for some
                $\e\in (0,1)$ and for all subsets $ \hat S
                \subseteq \hat{X} $ we have \hfill 
                \[ (1-\e)\rad^2(\hat S) \le
                \rad^2(\widehat{f(S)}) \le (1+\e)\rad^2(\hat S), \]
		then  
		the weighted \Cech filtrations of $ \hat X $ and $ f(\hat X) $
		are multiplicatively $ (1-\e)^{-1/2} $ interleaved.
	\end{lemma}

	\section{$\e$-Distortion maps preserve $k$-distance \Cech filtrations}
	\label{sec: results}

	For the subsequent theorems, we denote by $P$ a set of $n$ points in
        $\R^D$.

	Our first theorem shows that for the points in $P$, the
        pointwise $k$-distance $d_{P,k}$ is approximately preserved by a random subgaussian matrix satisfying Lemma~\ref{JLlemma}.

	\begin{theorem}\label{jlkdist} Given $\e \in \left(0,1\right]$, any $\e$-distortion map with respect to $P$ 
        $f \colon \mathbb{R}^D \to \mathbb{R}^d $, where $ d=O(\e^{-2}\log n)$ satisfies 
		for all points $ x \in P $:
		\[ (1-\e) d^2_{P,k}(x) \le d^2_{f(P),k}(f(x)) \le (1+\e)d^2_{P,k}(x). \]
	\end{theorem}

	\begin{proof}[Proof of Theorem~\ref{jlkdist}]
	The proof follows from the observation that the squared $k$-distance from any point $p \in P$ to the set $P$, is 
	a convex combination of the squares of the Euclidean distances to the $k$ nearest neighbours of $p$. 
	Since the mapping in the JL Lemma~\ref{JLlemma} is linear and $(1\pm \e)$-preserves squared pairwise distances, 
	their convex combinations also get $(1\pm\e)$-preserved. 
        \end{proof}

	As mentioned previously, the preservation of the pointwise $k$-distance does not imply the preservation of the \Cech complex formed using the points in $P$. 
        Nevertheless, the following theorem shows that this can always be done in dimension $O(\log n/\e^2)$.

	Let $ \wbary_{P,k} $ be the set of iso-barycenters of every $ k
        $-subset of $ P $, weighted as  
in
        Section~\ref{subsec:dtmbackground}. 
	Recall from Section~\ref{subsec:ph} that the weighted \Cech complex 
	$\check{C}_\alpha(\wbary_{P,k})$ is identical to the \Cech complex 
	$\check{C}_\alpha(P)$ for the $k$-distance. We now want to apply 
        Lemma~\ref{lotzlemmawtd}, for which the following theorem will be 
        needed.

	\begin{theorem}[$k$-distance]
		\label{persth1}
		Let $ \hat{\sigma} \subseteq \wbary_{P,k} $ be a
                simplex in the weighted \Cech complex $
                \check{C}_{\alpha}(\wbary_{P,k}) $. 
		  Then, given $d \leq D$ such that there exists a $\e$-distortion map $
                  \mathnormal{f}:\mathbb{R}^{D} \to \mathbb{R}^{d}$ with respect to $P$, 		
                it holds that 
		 \[(1-\e)\rad^2(\hat{\sigma}) \le \rad^2(\widehat{f(\sigma)}) \le (1+\e)\rad^2(\hat{\sigma}).\] 
              \end{theorem}
	
	\begin{proof}[Proof of Theorem \ref{persth1}]
	        Let $ \hat{\sigma} =
                \{\hat{b}_1,\hat{b}_2,...,\hat{b}_m\} $, where
                $\hat{b}_i$ is the weighted point defined in
                Section~\ref{subsec:ph}, i.e.  $\hat{b}_i=(b_i, w(b_i))$ with $b_i \in \bary_{P,k} $ and $w(b_i) = -\frac{1}{k}\sum_{l=1}^k \|b_i-p_{il}\|^2$, 
	        where $p_{i,1},\ldots,p_{i,k} \in P$ are such that $b_i = \frac{1}{k}\sum_{j=1}^k p_{i,j}$.
		Applying Lemma~\ref{l:meb-rad-conv-comb} to $\hat{\sigma}$, we have that 
		\begin{eqnarray}
			\rad^2(\hat{\sigma})   & = &\dfrac{1}{2}\sum_{i,j\in I} \lambda_i\lambda_j D(\hat{b}_i,\hat{b}_j). 
				               \label{eq:rad-sigma}
		\end{eqnarray}
                By Lemma~\ref{l:pow-dist-pairwise-dist}, the distance between $\hat{b}_i$ and $\hat{b}_j$ is $D(\hat{b}_i,\hat{b}_j) = 
                \frac{1}{k^2}\sum_{l,s=1}^k \|p_{i,l}-p_{j,s}\|^2$.
                As this last expression is a convex combination of squared pairwise distances of points in $P$, it is $(1\pm \e)$-preserved by any $\e$-distortion map
                with respect to $P$, 
                which implies that the convex combination $\rad^2(\hat{\sigma}) = \frac{1}{2}\sum_{i,j\in I} \lambda_i\lambda_j D(\hat{p}_i,\hat{p}_j)$ corresponding 
                to the squared radius of $\sigma$ in $\R^D$, will be $(1\pm \e)$-preserved.

		Let $f:\R^D\to \R^d$ be an $\e$-distortion map with respect to $P$, from $\R^D$ to $\R^d$, where $d$ will be chosen later. 
                By Lemma~\ref{l:meb-rad-conv-comb}, the centre of $\widehat{f(\sigma)}$ is a convex combination of the points $(f(b_i))_{i=1}^m$.
                Let the centre $c(\widehat{f(\sigma)})$ be given by $c(\widehat{f(\sigma)}) = \sum_{i\in I} \nu_i D(\widehat{f(b_i)})$.
                where for $i\in I$, $\nu_i\geq 0$, $\sum_i \nu_i =1$. Consider the convex combination of power distances 
                $\sum_{i,j\in I} \nu_i \nu_j D(\hat{b}_i,\hat{b}_j)$. 
                Since $f$ is an $\e$-distortion map with respect to $P$, by Lemmas~\ref{l:pow-dist-pairwise-dist} and~\ref{JLlemma} we get
                \begin{eqnarray}
                    \dfrac{1}{2}(1-\e)\sum_{i,j\in I} \nu_i\nu_j D(\hat{b}_i,\hat{b}_j) 
                                                    &\leq& \dfrac{1}{2}\sum_{i,j\in I} \nu_i\nu_j D(\widehat{f(b_i)},\widehat{f(b_j)}) \;\;=\;\; \rad^2(\widehat{f(\sigma)}).
                    \label{eqn:nu-lambda-jl}
                \end{eqnarray}
                On the other hand, since the squared radius is a minimizing function by definition, we get that 
                \begin{eqnarray}
                    \rad^2(\hat{\sigma}) &=& \dfrac{1}{2}\sum_{i,j\in I} \lambda_i\lambda_j D(\hat{b}_i,\hat{b}_j) 
                                             \;\;\leq\;\; \dfrac{1}{2}\sum_{i,j\in I} \nu_i\nu_j D(\hat{b}_i,\hat{b}_j) \label{eqn:lambda-nu-orig}\\
                                             &\leq& \dfrac{1}{(1-\e)}\rad^2(f(\sigma)), \mbox{ by~\eqref{eqn:nu-lambda-jl}}. \notag\\
                    \rad^2(\widehat{f(\sigma)}) &=& \dfrac{1}{2}\sum_{i,j\in I} \nu_i\nu_j D(\widehat{f(b_i)},\widehat{f(b_j)})  
                                                \;\;\leq\;\; \dfrac{1}{2}\sum_{i,j\in I} \lambda_i\lambda_j D(\widehat{f(b_i)},\widehat{f(b_j)}). \label{eqn:nu-lambda-im}
                \end{eqnarray}
                Combining the inequalities~\eqref{eqn:nu-lambda-jl}, ~\eqref{eqn:lambda-nu-orig},~\eqref{eqn:nu-lambda-im} gives
                \begin{eqnarray*}
                    (1-\e)\rad^2(\hat{\sigma}) &\leq& \rad^2(\widehat{f(\sigma)}) 
                                                      \;\;\leq\;\; \dfrac{1}{2}\sum_{i,j\in I} \lambda_i\lambda_j D(\widehat{f(b_i)},\widehat{f(b_j)})
                                                      \;\;\leq\;\; (1+\e)\rad^2(\hat{\sigma}).
                \end{eqnarray*}
                where the final inequality follows by Lemma~\ref{JLlemma}, since $f$ is an $\e$-distortion map with respect to $P$.
                Thus, we have that 
                \[ (1-\e)\rad^2(\hat{\sigma}) \;\;\leq\;\; \rad^2(\widehat{f(\sigma)}) \;\;\leq\;\; (1+\e)\rad^2(\hat{\sigma}) ,\]
                which completes the proof of the theorem.
              \end{proof}


\begin{theorem}[Approximate $ k $-distance]
\label{thm:approx-k-dist}
    Let $\hat{P}$ be the weighted points associated with $P$, introduced in
     Definition~\ref{def:approximation}
     (Equ.~\ref{eq:dppx-approx}). Let, in addition,  $
     \hat{\sigma} \subseteq \hat{P} $ be a simplex in the associated weighted \Cech
     complex $ \check{C}_{\alpha}(\hat{P}) $.   
Then an $\e$-distortion mapping with respect to $P$, $ \mathnormal{f}:\mathbb{R}^{D} \to
\mathbb{R}^{d} $ satisfies: $
(1-\e)\rad^2(\hat{\sigma}) \le \rad^2(\widehat{f(\sigma)}) \le
(1+\e)\rad^2(\hat{\sigma})$. 

\end{theorem}

	\begin{proof}[Proof of Theorem~\ref{thm:approx-k-dist}]

        Recall that, in Section \ref{subsec:dtmbackground}, we defined the approximate $ k $-distance to be 	
	$\tilde{d}_{P,k}(x) := \min_{p \in P}\sqrt{ D(x,\hat{p})}$,  where $\hat{p} = (p,w(p))$ is a weighted point, having weight $ w(p)= - d_{P,k}^2(p)$. So,  
	the \Cech complex would be formed by the intersections of the balls around the weighted points in $ P $. 
		The proof follows on the lines of the proof of
                Theorem~\ref{persth1}. Let $\hat\sigma=\{\hat{p}_1,\hat{p}_2,...,\hat{p}_m\}$, where $\hat{p}_1,\ldots,\hat{p}_m$ are weighted points in $\hat{P}$, 
                and let $c(\hat{{\sigma}})$ be the center of $\hat{\sigma}$.
		Applying again Lemma~\ref{l:meb-rad-conv-comb}, we get 
		\begin{eqnarray*}
			\rad^2(\hat{\sigma}) =
                        \frac{1}{2}\sum_{i,j\in I}\lambda_i\lambda_j\|p_i-p_j\|^2
                        + \sum_{i\in I}\lambda_iw(p_i)  =  \sum_{i,j\in I; i<j}\lambda_i\lambda_j\|p_i- p_j\|^2 +\sum_{i\in I}\lambda_iw(p_i), 
		\end{eqnarray*}
		where $ w(p)= d_{P,k}^2(p) $. In the second equality, we used the fact that the summand corresponding to a fixed pair of distinct indices $i<j$ is being 
		counted twice 
                and that the contribution of the terms corresponding to indices $i=j$ is zero.
		An $\e$-distortion map with respect to $P$ preserves pairwise distances
                and the $ k $-distance in dimension $ O(\e^{-2}\log n)
                $. The result then follows as in the proof of 
		Theorem \ref{persth1}.
	\end{proof}
	
Applying Lemma~\ref{lotzlemmawtd} to the theorems~\ref{persth1} and ~\ref{thm:approx-k-dist},
we get the following corollary.
\begin{corollary}
\label{cor:pers-homol-interleave}
	The persistent homology for the \Cech filtrations of $ P $ and its image $ f(P) $ under any $\e$-distortion mapping with respect to $P$, using the 
        $(i)$ exact $k$-distance, as well as the $(ii)$ approximate $ k $-distance, 
	are preserved upto a multiplicative factor of $(1-\e)^{-1/2}$. 
\end{corollary}

However, note that the approximation in Corollary~\ref{cor:pers-homol-interleave} $(ii)$ is with respect to the \emph{approximate} $k$-distance, 
which is itself an approximation of the $k$-distance by a distortion factor $3\sqrt{2}$, (i.e. bounded away from $1$ -- see~\eqref{eqn:approx-k-dist-approx}).

%

        \section{Extensions}
        \label{sec:appln}

        As Theorem~\ref{persth1} applies to arbitrary $\e$-distortion maps, it naturally follows that many of the extensions and variants of the 
        JL Lemma, e.g. discussed in Section~\ref{subsec:rand-proj}, have their corresponding versions for the $k$-distance as well. 
        In this section we elucidate some of the corresponding extensions of Theorem~\ref{persth1}. 
        These can yield better bounds for the dimension of the embedding, stronger dimensionality reduction results, or easier to implement reductions 
        in their respective settings.

        The first result in this section, is for point sets contained in a region of bounded Gaussian width.



        \begin{theorem} 
        \label{thm:bdd-gauss-width} 
        Let $P \subset \R^D$ be a finite set of points, 
        and define $S := \{(x-y)/\|x-y\|\;:\; x,y\in P\}$. Let $w(S)$ denote the Gaussian width of $S$. Then, 
        given any $\e,\delta \in (0,1)$, 
        any subgaussian $\e$-distortion map from $\R^D$ to $\R^d$ 
        preserves the persistent homology of the $k$-distance based 
        \Cech filtration associated to $P$, up to a multiplicative factor of $(1-\e)^{-1/2}$,
        given that 
        \[d \geq \frac{\pth{w(S)+\sqrt{2\log (2/\delta)}}^2}{\e^2}+1.\]
        \end{theorem}

        Note that the above theorem is not stated for an arbitrary $\e$-distortion map. 
        Also, since the Gaussian width of an $n$-point set is at most $O(\log n)$ (using e.g. the Gaussian concentration inequality, 
        see e.g.~\cite[Section 2.5]{DBLP:books/daglib/0035704}), Theorem~\ref{thm:bdd-gauss-width}
        strictly generalizes Corollary~\ref{cor:pers-homol-interleave}. \\ 

        \begin{proof}[Proof of Theorem~\ref{thm:bdd-gauss-width}]
        By Theorem~\ref{thm:lotz-gordon-gw}, the scaled random Gaussian matrix $f:x\mapsto  \pth{\sqrt{D/d}}Gx$ is an $\e$-distortion map
        with respect to $P$, having 
        dimension $d\geq \frac{\pth{w(S)+\sqrt{2\log (2/\delta)}}^2}{\e^2}+1$. 
        Now applying Theorem~\ref{persth1} to the point set $P$ with the mapping $f$, immediately gives us that  
        for any simplex $\hat{\sigma} \in \check{C}_{\alpha}(\hat{\bary}_{P,k})$, where $\check{C}_{\alpha}(\hat{\bary}_{P,k})$ 
        is the weighted \Cech complex with parameter $\alpha$, the squared radius $\rad^2(\hat{\sigma})$ is preserved up to a 
        multiplicative factor of $(1\pm \e)$. By Lemma~\ref{lotzlemmawtd}, this implies that the persistent homology for the 
        \Cech filtration is $(1-\e)^{-1/2}$-multiplicatively interleaved.
        \end{proof}

        For point sets lying on a low-dimensional submanifold of a high-dimensional Euclidean space, one can obtain an embedding having smaller dimension, using the 
        bounds of Baraniuk and Wakin~\cite{DBLP:journals/focm/BaraniukW09} or Clarkson~\cite{DBLP:conf/compgeom/Clarkson08}, 
        which will depend only on the parameters of the submanifold.   

        \begin{theorem}
        \label{thm:manif-bdd}
        There exists an absolute constant $c>0$ such that, given a finite point set $P$ lying on a connected, compact, orientable, 
        differentiable $\mu$-dimensional submanifold $M \subset \R^D$, and $\e,\delta \in (0,1)$, an $\e$-distortion map 
        $f:\R^D\to \R^d$ preserves the persistent homology of the \Cech filtration computed on $P$, using the $k$-distance, 
        provided 
        \[ d \geq c\pth{\frac{\mu\log(1/\e)+\log(1/\delta)}{\e^2} + \frac{C(M)}{\e^2}},\]
        where $C(M)$ depends only on $M$. 
        
        \end{theorem}

        \begin{proof}[Proof of Theorem]
         The proof follows directly, by applying the map in Clarkson's bound (Theorem~\ref{thm:clark-manif-bdd}) as the $\e$-distortion map in Theorem~\ref{persth1}. 
        \end{proof}
	
        Next, we state the terminal dimensionality reduction version of Theorem~\ref{persth1}. This is a useful result when 
        we wish to preserve the distance (or $k$-distance) from any point in the ambient space, to the original point set.
        \begin{theorem}
        \label{thm:term-kdist-dim-red}
        Let $P \in \R^D$ be a set of $n$ points. Then, given any $\e\in (0,1]$, there exists a map $f:\R^D\to \R^d$, where 
        $d = O\pth{\frac{\log n}{\e^2}}$, such that the persistent homology of the $k$-distance based \Cech filtration associated 
        to $P$ is preserved up to a multiplicative factor of $(1-\e)^{-1/2}$, and the $k$-distance of any point in $\R^D$ to 
        $P$, is preserved up to a $(1\pm O(\e))$ factor.
        \end{theorem}
	
        \begin{proof}
        The second part of the theorem follows immediately by applying Theorem~\ref{thm:nel-nar-term-dim-red}, with the point set $P$ as the set of terminals.
        By Theorem~\ref{thm:nel-nar-term-dim-red} $(ii)$, the dimensionality reduction map of~\cite{DBLP:conf/stoc/NarayananN19} 
        is an outer extension of a subgaussian $\e$-distortion map $\Pi:\R^D\to \R^{d-1}$. Now applying Theorem~\ref{persth1} to $\Pi$ gives 
        the first part of the theorem.
        \end{proof}

	\section{Conclusion and Future Work}
	\label{sec:future-work}
\vspace{2mm}

	\noindent{\bf $k$-Distance Vietoris-Rips and Delaunay filtrations.}  
        Since the Vietoris-Rips filtration~\cite[Chapter 4]{DBLP:books/daglib/0041250}
        depends only on pairwise distances, it follows from Theorem~\ref{jlkdist} 
        that this filtration with $k$-distances, is preserved upto a multiplicative factor of $(1-\e)^{-1/2}$, under a 
        Johnson-Lindenstrauss mapping.
        Furthermore, the $k$-distance Delaunay and the \Cech filtrations~\cite[Chapter 4]{DBLP:books/daglib/0041250} 
        have the same persistent homology. Corollary~\ref{cor:pers-homol-interleave} $(i)$ therefore implies that the $k$-distance Delaunay 
        filtration of a given finite point set $P$ is also $(1-\e)^{-1/2}$-preserved under an $\e$-distortion map with respect to $P$. Thus, 
        Corollary~\ref{cor:pers-homol-interleave} $(ii)$ 
        apply also to the approximate $k$-distance Vietoris-Rips and $k$-distance Delaunay filtrations.
\vspace{2mm}

	\noindent {\bf Kernels.}  
Other distance functions defined using kernels 
have proved successful in overcoming issues due to 
	outliers. Using a result analogous to Theorem~\ref{jlkdist}, we
        can show that random projections preserve the persistent homology for kernels 
	up to a $C(1-\e)^{-1/2}$ factor where $C$ is a constant. We
        don't know if we can make $C=1$ as for the $k$-distance.
\vspace{2mm}
	
\section*{Acknowledgement} The authors would like to thank the referees for their comments and observations, which helped add more consequences to the 
main results and also improved the quality of the paper. 

\section*{Statement on Conflict of Interest:} On behalf of all authors, the corresponding author states that there is no conflict of interest.

\bibliographystyle{plainurl}
	\bibliography{bibliography} 

\section*{Appendix}
The following elementary lemma gives identity~\eqref{eqn:basic-var-id}.

          \begin{lemma}   \label{l:basic-var-id}
          Let $b = \sum_{i=1}^k \lambda_i p_i$ be a convex combination of points $p_1,\dots,p_k$. Then for any point $x\in \R^D$,
          \begin{eqnarray}
             \sum_{i=1}^k \lambda_i \|x-p_i\|^2  &=& \|x-b\|^2 + \sum_{i=1}^k \lambda_i \|b-p_i\|^2. 
          \end{eqnarray}
          \end{lemma}

        \begin{proof}
           Recall the following fundamental relation between the variance and expectation of a random variable.
           Let $X\in \R^D$ be a random variable bounded in $\ell_2$. Then, by one characterization of the variance,
           \begin{eqnarray}
               \Ex{\|X\|^2} &=& \mathrm{Var}(X) + \|\Ex{X}\|^2 . \label{eqn:var-exsq-id}
           \end{eqnarray}
              Consider a point $x\in \R^D$, a set $P \subset \R^D$ of $k$ points, and a probability distribution 
              $\{\lambda_i\}_{i=1}^k$, along with a weighted sum $b= \sum_{i=1}^k \lambda_i p_i$. The random vector $Y$ supported on $P$, 
              with probability $\Prob{Y=p_i}=\lambda_i$, then satisfies $\Ex{Y}=b$. Define $X:=x-Y$, so that $\Ex{X} = x-b$. Then
              \begin{eqnarray*}
              \Ex{\|X\|^2} &=& \sum_{i=1}^k \lambda_i\|x-p_i\|^2, \mbox{ and} \\
              \mathrm{Var}(X) &=& \Ex{\| X-\Ex{X}\|^2} = \Ex{\|Y-b\|^2} = \sum_{i=1}^k \lambda_i \|b-p_i\|^2.
              \end{eqnarray*}
              Substituting in~\eqref{eqn:var-exsq-id}, the claim follows.
        \end{proof}
	
        \begin{lemma} [Lotz~\cite{doi:10.1098/rspa.2019.0081}, Lemma 4.2]
        \label{l:simp-fact}
           Given a set $P=\{p_1,\ldots,p_l\}\subset \R^D$ of points, and a point $c\in \R^D$ such that $c=\sum_{i\in I}\lambda_i p_i$, where 
           $I$ is a subset of indices from $[l]$, and $\lambda_i\geq 0$, with $\sum_{i\in I}\lambda_i=1$.
           Then 
           \[ \sum_{i\in I} \lambda_i \|c-p_i\|^2 = \frac{1}{2}\sum_{i,j\in I}\lambda_i\lambda_j\|p_i-p_j\|^2.\]
           Consequently, 
           \[ \rad^2(P) = \sum_{i\in I} \lambda_i \|c-p_i\|^2 = \frac{1}{2}\sum_{i,j\in I}\lambda_i\lambda_j\|p_i-p_j\|^2.\]
           where $\rad^2(P)$ is the squared radius of the minimum enclosing ball of the set $P$ of points.
        \end{lemma}

        \begin{proof}
           The proof again follows directly from Eqn.~\eqref{eqn:var-exsq-id}. Suppose we choose two random points $X_1,X_2$ independently from $P$, with the point $p_i$ 
         being chosen with probability $\lambda_i$. Then, 
         \[ \Ex{\|X_1-X_2\|^2} = \|\Ex{X_1-X_2}\|^2 + Var(X_1-X_2).\]
         Evaluating, we get that 
         \begin{eqnarray*}
            \Ex{\|X_1-X_2\|^2} &=& \sum_{i,j\in I} \lambda_i\lambda_j\|p_i-p_j\|^2, \\
            \|\Ex{X_1-X_2}\|^2   &=& \|\Ex{X_1}-\Ex{X_2}\|^2 = 0, \mbox{ and} \\ 
            Var(X_1-X_2)       &=& 2\cdot Var(X_1) = 2\sum_{i\in I}\lambda_i\|c-p_i\|^2,
         \end{eqnarray*}
         where in the last line, we used the fact that $X_1$ and $X_2$ are independent. Substituting the above values in the variance identity~\eqref{eqn:var-exsq-id}, 
         completes the proof.
        \end{proof}

        A probabilistic proof of Lemma~\ref{l:pow-dist-pairwise-dist} is also provided below.

        \begin{proof}[Lemma~\ref{l:pow-dist-pairwise-dist} - a probabilistic proof]
        Consider the following random experiment: pick a random point $X$ from $p_1,\ldots, p_k$ according to the distribution $(\lambda_i)_{i=1}^k$ 
        and another independently random point $Y$ from $q_1,\ldots,q_k$ according to $(\mu_i)_{i=1}^k$. 

        Using the law of total variance on the variable $X-Y$, conditioning on $Y$, we get that 
        \[ Var(X-Y) = Var_{Y}(\mathbb{E}_X(X-Y|Y)) + \mathbb{E}_Y[Var_X(X-Y|Y)] \mbox{ or,}\]
        \begin{eqnarray}
           \Ex{\|X-Y\|^2} &=& \|\Ex{X-Y}\|^2 + Var_{Y}(\mathbb{E}_{X}(X-Y|Y)) + \mathbb{E}_{Y}[Var_X(X-Y|Y)]  \label{eqn:cond-var-id}
        \end{eqnarray}

        Let us consider the terms in the above equation one by one. 
        \begin{enumerate}
           \item The LHS has $\Ex{\|X-Y\|^2}$, which by the independence of $X$ and $Y$ is clearly equal to 
           $\sum_{i,j=1}^k \lambda_i\mu_j \|p_i-q_j\|^2$.
           \item In the RHS, the first term is $\|\Ex{X-Y}\|^2=\|b_1-b_2\|^2$. 
           \item The second term is $Var_{Y}(\mathbb{E}_{X}(X-Y)|Y)$, which is equal to 
        $Var_{Y}(b_1-Y) = Var(Y) = \sum_{i=1}^k \mu_i\|b_2-q_i\|^2$, where the last expression was evaluated directly 
        from the definition of variance, i.e. $Var(Z) = \Ex{(Z-\Ex{Z})^2}$, and that for constant $a$, $Var(a-Z)=Var(Z)$. 
           \item The final term is $\mathbb{E}_{Y}(Var_X(X-Y|Y))$. Conditioning on $Y$, the variance $Var_X(X-Y|Y)=Var(X)$, i.e. $\sum_{i=1}^k \lambda_i \|b_1-p_i\|^2$. Since 
        this holds for each value of $Y$, we get that $\mathbb{E}_Y[Var_X(X-Y|Y)] = \mathbb{E}_Y[Var(X)] = \sum_{i=1}^k \lambda_i \|b_1-p_i\|^2$. 
        \end{enumerate}
        Substituting the above expressions for the terms in~\eqref{eqn:cond-var-id}, we get 
        \[ \sum_{i,j=1}^k \lambda_i\mu_j \|p_i-q_j\|^2 = \|b_1-b_2\|^2 + \sum_{i=1}^k \mu_i \|b_2-q_i\|^2 + \sum_{i=1}^k \lambda_i \|b_1-p_i\|^2. \]
        \end{proof}

\end{document}